\documentclass[11pt]{article}
\usepackage{amsmath,amsthm,amssymb}
\usepackage[numbers]{natbib}
\usepackage[margin=1in]{geometry}
\usepackage{thm-restate}
\usepackage{comment}
\usepackage{microtype}
\usepackage[colorlinks,linkcolor=blue,citecolor=purple]{hyperref} 

\usepackage{tikz}
\usetikzlibrary{shapes}
\usetikzlibrary{patterns}
\usetikzlibrary{calc,3d,intersections, positioning,intersections,shapes}
\usetikzlibrary{arrows,positioning,decorations.pathmorphing,trees}
\usetikzlibrary{decorations.fractals}
\usepackage{tikz-3dplot}
\usetikzlibrary{arrows.meta}
\tikzset{
  MyStealth/.tip={Stealth[line width=.7pt, inset=1pt, length=10pt, angle'=30]}
}

\tikzstyle{n}=[]
\tikzstyle{bl-ci}=[circle, draw, fill=black, inner sep=0pt, minimum width=4pt]
\tikzstyle{g-ci}=[circle, draw, fill=white, inner sep=0pt, minimum width=4pt]
\tikzstyle{cl-sq}=[rectangle, minimum height=6pt, draw, fill=white, inner sep=0pt, minimum width=6pt]
\tikzstyle{bl-sq}=[rectangle, minimum height=6pt, draw, fill=black, inner sep=0pt, minimum width=6pt]
\tikzstyle{r-sq}=[rectangle, minimum height=6pt, draw, fill=black, inner sep=0pt, minimum width=6pt]
\tikzstyle{b-sq}=[rectangle, minimum height=6pt, draw, fill=gray, inner sep=0pt, minimum width=6pt]
\usepackage{enumitem}
\usepackage[ruled,vlined]{algorithm2e}
\usepackage{authblk}

\usepackage{xcolor}
\usepackage{todonotes}

\newtheorem{thm}{Theorem}

\newtheorem{lem}[thm]{Lemma}

\newtheorem{claim}[thm]{Claim}

\newtheorem*{remark*}{Remark}

\usepackage[cm]{fullpage}

\begin{document}

\title{The Condorcet Dimension of Metric Spaces}

\author[1]{Alexandra Lassota} 
\author[2]{Adrian Vetta}
\author[3]{Bernhard von Stengel}

\affil[1]{Eindhoven University of Technology: {\tt a.a.lassota@tue.nl}}
\affil[2]{McGill University: {\tt adrian.vetta@mcgill.ca}}
\affil[3]{London School of Economics: {\tt b.von-stengel@lse.ac.uk}}

\date{}

\maketitle

\bibliographystyle{plainnat}

\renewcommand{\thefootnote}{\arabic{footnote}}

\begin{abstract}
A Condorcet winning set is a set of candidates such that no other candidate is preferred by at least half the voters over
all members of the set.
The Condorcet dimension, which is the minimum cardinality of a Condorcet winning set, is known to be at most logarithmic in the number of candidates. 
We study the case of elections where voters and candidates
are located in a $2$-dimensional space with preferences 
based upon proximity voting.
Our main result is that the Condorcet dimension is at most $4$, under both the Manhattan norm and the infinity norm, which are natural measures in electoral systems.
We also prove that any set of voter preferences can be embedded into a metric space of sufficiently high dimension for any $p$-norm, including the Manhattan and infinity norms.
\end{abstract}

\section{Introduction}\label{s:introduction}

Consider an \textit{election}, that is, a set of candidates $\mathcal{C}$
of cardinality $m$ and a set of voters $\mathcal{V}$ of
cardinality~$n$, where each voter $v\in \mathcal{V}$ has
strict preferences $\succ_v$ over all candidates.
An ideal winner in such an election is a Condorcet winner, a candidate 
that ``beats" any other candidate in a pairwise voting contest.
Formally, $i\in C$ is a {\em Condorcet winner} if, for every candidate $j\in \mathcal{C}\setminus \{i\}$, more than
half of the voters prefer $i$ over~$j$. To avoid ties, we assume the number~$n$ of voters is odd. Unfortunately, it is easy to construct elections where a Condorcet winner does not exist; indeed typically there is no Condorcet winner.

Given this, Elkind, Lang and Saffidine~\cite{ELS15} proposed a relaxation where,
rather than a single winning candidate, we desire a set of winning candidates that collectively beats any other candidate.
Specifically, a set of candidates $S\subseteq \mathcal{C}$ is a {\em Condorcet winning set} if, 
for every candidate $j\in \mathcal{C}\setminus S$, more than half of the voters prefer some (voter-dependent) candidate in $S$ to $j$. 
That is,
\begin{equation}\label{eq:CWS}
\forall j\in \mathcal{C}\setminus S, \  \#\{v\in \mathcal{V} : \exists i\in S, v \text{ prefers } i \text{ over } j \} > \tfrac12 n
\end{equation}
where $\#$ denotes set cardinality.
Elkind et al.~\cite{ELS15} showed that a Condorcet winning set always exists provided we allow it to be of logarithmic size in~$m$. 
We give the short proof, which highlights the distinction between the graph interpretation of the problem and Condorcet winning sets. 

To do so, though, we first need some basic graph definitions.
A {\em directed graph} $G=(V,A)$, or {\em digraph}, consists of a set $V$ of elements, called {\em vertices}, and a set $A$ of ordered pairs of vertices, called {\em arcs}. For each arc $a=(i,j)\in A$, the vertex $i$ is the {\em tail} of $a$ and the vertex $j$ is the {\em head} of~$a$.
Graphically, the arc $(i,j)$ is visualized as an arrow pointing from $i$ to $j$ and we say that $i$ is an {\em in-neighbour} of $j$ and that $j$ is an {\em out-neighbour} of $i$. The {\em in-degree} and {\em out-degree} of a vertex are the total number of its in-neighbours and out-neighbours, respectively.
Finally, in a directed graph $G=(V,A)$ a {\em dominating set} $T$ is a subset of vertices with the property that, for every vertex $i\in V\setminus T$, there exists $j\in T$ such that $(i, j)$ is an arc in $A$.

\begin{thm}[\cite{ELS15}]\label{thm:log}
In any election, there is a Condorcet winning set of size at most $\lceil\log m\rceil$.
\end{thm}
\begin{proof}
Consider the {\em majority digraph} which contains a vertex for each candidate $i\in \mathcal{C}$ and an arc~$(j,i)$ if a majority of voters prefer $i$ over $j$. Since the number of voters is odd, the majority digraph
is a \textit{tournament} (a digraph with exactly one arc between any two nodes). 
Megiddo and Vishkin~\cite{MV88} proved that any tournament contains a dominating set of size
at most $\lceil\log m\rceil$. To see this, simply select the vertex with the largest number of in-neighbours; remove this vertex and its in-neighbours and recurse. 
The selected set of vertices $T$ is a dominating set in the majority digraph and, thus, satisfies the property 
\begin{equation}\label{eq:dom-set}
\forall j\in \mathcal{C}\setminus T, \  \exists i\in T, \ \#\{v\in \mathcal{V} :  v \text{ prefers } i \text{ over } j \} > \tfrac12 n
\end{equation}
Observe that (\ref{eq:dom-set}) is a stronger property than~(\ref{eq:CWS}). Consequently, $T$ is a Condorcet winning set of size
at most $\lceil\log m\rceil$.
\end{proof}
In this paper, we study the {\em Condorcet dimension}, the minimum cardinality of a Condorcet winning set. Theorem~\ref{thm:log} gives a logarithmic upper bound on the Condorcet dimension.
In terms of lower bounds, Elkind et al.~\cite{ELS15} conjectured that elections exist with arbitrarily high Condorcet dimension.
Since a Condorcet winner need not exist, there are elections with Condorcet dimension at least $2$. Further, Elkind et al.~\cite{ELS15} presented instances with Condorcet dimension $3$; in addition, Geist~\cite{Gei14} presented an instance with just six candidates and six voters of Condorcet dimension~$3$.

However, no instances with Condorcet dimension greater than $3$ are currently known and we conjecture the Condorcet dimension is always at most $3$. The main results we prove in this paper are that the Condorcet dimension is at most $4$ for two special cases of elections under the spatial model of voting.

We remark that, motivated by our work and conjecture, Charikar et al.~\cite{CLR25} subsequently observed that
the groundbreaking research on approximately stable committee selection by Jiang, Munagala and Wang~\cite{JMW20} implies a constant upper bound of $16$ on the Condorcet dimension in any election. 
Refining the techniques of~\cite{JMW20}, Charikar et
al.~\cite{CLR25} show how probabilistic sampling based upon
the minimax theorem implies the existence of a Condorcet
winning set of size $6$.
This has recently been improved to $5$ by Nguyen et al.~\cite{nguyen2025}.
Thus, the Condorcet dimension is at most $5$ in any election.

\subsection{The Spatial Model of Voting}
Spurred by the conjecture that the Condorcet dimension is at most three in any election, we study Condorcet winning sets in the classical
{\em spatial model of voting} where the voters and candidates are located in a metric space. We make the standard assumption of {\em proximity voting}, whereby each voter $v$ ranks the candidates based upon their distance to $v$, with closer candidates more preferred.
We assume strict preferences; specifically, we assume the candidates are in general position.

The spatial model of voting originates in the seminal works of Hotelling~\cite{Hot29} and Black~\cite{Black48}, and was
formalized by Downs~\cite{Downs57}. These works all considered 
one-dimensional metric spaces.
Of particular interest here is the {\em median voter theorem}, 
where Black~\cite{Black48} showed that the closest candidate to the 
median voter is a Condorcet winner. 
The study of higher dimensional metric spaces was instigated by
Davis, Hinich, and Ordeshook~\cite{DHO70}, where Condorcet winners exist
only under very restrictive conditions~\cite{DDH72}.
Furthermore, the absence of a Condorcet winner has serious consequences in terms of ``agenda manipulation". For example, in this situation, the McKelvey-Schofield Chaos Theorem~\cite{McK76, Sch78} states that, starting from any proposed policy (candidate), any other policy can be implemented (elected) after a sequence of pairwise majority votes, given the addition of an appropriate set of intermediate policies for these votes.

The spatial model has since become ubiquitous in the study of voter and candidate behaviour, as illustrated by the books~\cite{EH84, EH90, MMG99, Poole05, Sch08}.
The spatial model has also recently attracted huge interest in computer science~\cite{ABP15, AP17, BLS19, CRW24, FFG16, GHS20, SE17}.
In particular, in computational social choice there has been much focus on the impact of spatial voting on 
impossibility theorems, on the computational complexity of voting rules, and on informational aspects such as metric distortion; for details see the surveys given in~\cite{AFS21, ELP25}.

In spatial models of dimension greater than one, a choice of distance measure can be made. In this paper, we focus on the 
$p$-norms, that is, for ${\bf z}\in \mathbb{R}^D$, the $p$-norm is defined as $||{\bf z}||_p = \sqrt[p]{z_1^p+z_2^p+\cdots +z_D^p}$. Specifically, we study the Manhattan norm ($p=1$) and the infinity norm ($p=\infty$). These two norms are most appropriate in political settings where each axis represents a distinct policy or characteristic~\cite{Egu11, ELP25}.
Under the Manhattan norm, a voter prefers the candidate whose sum of policy differences over all policies is minimized;
that is, the voter desires a candidate close to them on {\em average} over all policies.
Under the infinity norm, a voter prefers the candidate whose maximum difference over all policies is minimized; that is, the voter desires a candidate close to them on {\em every} policy.

It is known that any set of voter preferences over candidates has a corresponding embedding in a metric space with proximity voting, for
the Euclidean norms~\cite{BL07} and Manhattan norms~\cite{CNS22}
However, this embedding may require the dimension $D$ of the metric space to be high, specifically, $D = \min\, [m-1, n]$.
We provide simple constructions that extends this conclusion to any $p$-norm,
provided $D = \min\, [m, n]$.

This is in stark contrast to practice where the embedding is often in small dimension $D$;
for example,~\cite{AG15} suggests that a small number of
dimensions provide a good explanation for German electoral data.
This observation motivates our study of low dimensional
metric spaces, in particular $D=2$.
Of course, if the dimension $D$ of the metric space is bounded
then this does restrict the set of feasible preferences in the election.
Indeed we show that this restriction guarantees that the Condorcet dimension is at most $4$ when $D=2$.

\subsection{Our Results}
We study metric spaces with bounded dimension $D$. 
Recall the median value theorem tell us that for $D=1$
a Condorcet winner always exists. What about the case $D=2$?
In fact, Davis et al.~\cite{DHO70} proved over half a century ago that a Condorcet winner need not exist in two dimensional metric spaces for the Euclidean norm.
Consequently, it makes sense to study the Condorcet dimension
for $D\ge 2$. We prove the following:
\begin{itemize}
    \item In a $2$-dimensional metric space (with Manhattan norm or infinity norm), the Condorcet dimension of any election is indeed at most $4$ (Theorem~\ref{thm:DS4} and Theorem~\ref{thm:four-infinity}).
    \item In a $2$-dimensional metric space under the $p$-norm, there exist instances with Condorcet dimension at least~$2$, for any $p\ge 1$ (Lemma~\ref{lem:LB}).
    \item Given $m$ candidates and $n$ voters,
an embedding of dimension $D= \min\, [m, n]$ can be computed in polynomial time, for any $p$-norm (Theorem~\ref{thm:simplex-m} and Theorem~\ref{thm:simplex-n}).
\end{itemize}

Our main technical result, presented in Sections
\ref{sec:manhattan} and \ref{sec:infinity}, is that for two
dimensions the Condorcet dimension is at most~$4$.
In Section~\ref{sec:lower}, we furthermore provide a lower bound of $2$, for any $p$-norm, by formalizing and instantiating the construction of~\cite{DHO70}.
Finally, in Section~\ref{sec:embed} we extend the conclusions of \cite{BL07} and \cite{CNS22} to any $p$-norm by showing
that any set of voter preferences over candidates has a corresponding embedding in a metric space with proximity voting with $D= \min\, [m, n]$.

\section{An Upper Bound on the Condorcet Dimension: Manhattan Norm}\label{sec:manhattan}

We begin with the Manhattan norm, whose distance
function is the $1$-norm. 
Thus, given two points~${\bf p}_1=(x_1, y_1)$ and ${\bf p}_2=(x_2, y_2)$, the distance between 
them is $d^{1}({\bf p}_1, {\bf p}_2) = |x_1-y_1| + |x_2- y_2|$.
Throughout this section, all notions of distance and closeness refer to the Manhattan norm; thus, we will omit the superscript and write $d({\bf p}_1, {\bf p}_2)=d^{1}({\bf p}_1, {\bf p}_2)$.

Let $\bar{x}$ be the median $x$-coordinate of the $n$ voters. Similarly, let $\bar{y}$ be the median $y$-coordinate of the voters. Without loss of generality, we may assume $\bar{x}=0$ and $\bar{y}=0$ by shifting all points accordingly.
The $x$-axis 
and the $y$-axis 
divide the plane into four closed quadrants $\{Q_1, Q_2, Q_3, Q_4\}$.
We label the quadrants in standard counter-clockwise order starting
with the all-nonnegative quadrant $Q_1$ (thus, $Q_3$ is the
all-nonpositive quadrant).

For each quadrant $Q_i$, $1\le i \le 4$, let $c_i$ be the candidate in $Q_i$ closest to the origin $(0,0)$. 
If $Q_i$ contains no candidates then $c_i$ does not exist; 
in this case, we define $c_i$ as a null element.
Otherwise $c_i$ exists and is unique by the assumption that the candidates are in general position.
However, at least one of the $c_i$ exists as there are $m\ge 1$ candidates.
The $c_i$ need not be distinct; for example, if $c^*=(0,0)$ is a candidate then $c_1=c_2=c_3=c_4=c^*$.
We claim that $S=\{c_1, c_2, c_3, c_4\}$, where $c_i$ is omitted from $S$ if it does not exist, is a Condorcet winning set.

\begin{lem}\label{lem:quadrant}
If $Q_i$ contains at least one candidate then at least half of all voters prefer $c_i$ in $Q_i$ over any other $c$ in~$Q_i$.
\end{lem}
\begin{proof}
Take any candidate $c = (x, y)$.
Without loss of generality, let $c$ be in the positive quadrant $Q_1$.
Then $c_1=(x_1, y_1)\in Q_1$ exists.
We may assume $c\neq c_1$. Since $c_1$ is the closest candidate to the origin in the positive 
quadrant $Q_1$, we have that $x+y \ge x_1+y_1$. Equivalently,
$x-x_1 \ge y_1-y$. 
We consider two cases (see Figure~\ref{fig:lem2}):
\begin{enumerate}
\item[(a)] $ x-x_1\ge y-y_1$. 
We will show that at least half the voters prefer $c_1=(x_1, y_1)$ over $c= (x, y)$.
Since $\bar{x}=0$ is the median $x$-coordinate, at least half the
voters have an $x$-coordinate at most zero. Take any such voter $v=(x_v, y_v)$ with $x_v\le 0$,
that is, any voter $v\in Q_2\cup Q_3$.
We claim $d(c,v) \ge d(c_1, v)$, that is, $v$ prefers $c_1$ over $c$; in fact, the inequality will be strict here as the candidates are in general position.
First, let $v\in Q_3$, that is, $x_v\le 0$ and $y_v\le 0$.
Then the following are equivalent:
\begin{eqnarray*}
d(c,v) &\ge& d(c_1,v)\\
|x-x_v|+ |y-y_v| &\ge&|x_1-x_v|+|y_1-y_v| \\
x+|x_v|+ y+|y_v| &\ge&x_1+|x_v|+y_1+|y_v| \\
x+ y &\ge&x_1+y_1
\end{eqnarray*}
which holds by the choice of $c_1$.
Second, let $v\in Q_2$, that is, $x_v\le0$ and $y_v\ge0$.
Now $x-x_1\ge y-y_1$ by the case assumption.
Furthermore, $x-x_1\ge y_1-y$
by the choice of $c_1$. Therefore, $x-x_1\ge |y_1-y|$.
Then we have the equivalences
\begin{eqnarray*}
d(c,v) &\ge& d(c_1,v)\\
 |x-x_v|+ |y-y_v| &\ge&|x_1-x_v|+|y_1-y_v| \\
x+|x_v|+ |y-y_v| &\ge&x_1+|x_v|+|y_1-y_v| \\
x-x_1+ |y-y_v| &\ge&|y_1-y_v|
\end{eqnarray*}
where the latter holds because
\[
x-x_1+ |y-y_v| \ge
|y_1-y|+|y-y_v|\ge
|y_1-y_v|
\]
by the triangle inequality.
Consequently, voter $v$ prefers candidate $c_1$ over candidate $c$,
as desired.
Thus, any voter $v\in Q_2\cup Q_3$ prefers $c_1$ over~$c$ and these two quadrants contain at least half of the voters. An illustration of this proof is given in Figure~\ref{fig:lem2}(a).

 \begin{figure}[ht]
 \begin{tikzpicture}[scale=.7,>={MyStealth}]
 \node []  at (-5.5,4.5) {(a)};            
 \draw[step=.5cm,color=gray, ultra thin] (-4.2,-4.2) grid (4.2,4.2);
\draw[->,black, ultra thick] (0,-4.3) -- (0, 4.3) node[above]{\footnotesize$y$};
  \draw[->,black, ultra thick] (-4.3,0) -- (4.3, 0) node[right]{\footnotesize$x$};
          \draw[black, dashed, thick] (0.5,0) -- (4.2, 3.7);
     \node [r-sq] at (2,1.5) {};
 \node [b-sq] at (3,1) {}; \node []  at (3.3,1.3) {{{\footnotesize $c$}}};
  \node [cl-sq] at (0.5,3.5) {}; 
   \node [cl-sq] at (3.5,2) {}; 
    \node [cl-sq] at (1.5,3) {}; 
        \node [cl-sq] at (2.5,4) {}; 
  \node [g-ci] at (-2,-2) {}; 
    \node [g-ci] at (-1,-3.5) {}; 
      \node [g-ci] at (-3,-1) {}; 
        \node [g-ci] at (-2.5,-4) {}; 
          \node [g-ci] at (-1,-1.5) {}; 
            \node [g-ci] at (-3.5,-3) {}; 
     \node [bl-ci] at (2,-1) {}; 
       \node [bl-ci] at (1,-3.5) {}; 
         \node [bl-ci] at (3.5,-2) {}; 
           \node [bl-ci] at (1.5,-0.5) {}; 
             \node [bl-ci] at (2.5,-2.5) {}; 
       \node [bl-ci] at (0.5,1) {}; 
       \node [bl-ci] at (1.5,2.5) {}; 
         \node [bl-ci] at (3.5,4) {}; 
           \node [bl-ci] at (2.5,0.5) {}; 
                    \node [bl-ci] at (2,3.5) {}; 
           \node [bl-ci] at (4,2.5) {};
                 \node [g-ci] at (-2.5,1.5) {}; 
       \node [g-ci] at (-1.5,3.5) {}; 
         \node [g-ci] at (-1,1) {}; 
           \node [g-ci] at (-3,2) {}; 
            \node [g-ci] at (-3.5,4) {}; 
\node []  at (4.3,4.3) {{\footnotesize $Q_1$}};
\node []  at (4.3,-4.3) {{\footnotesize $Q_4$}};
\node []  at (-4.3,-4.3) {{\footnotesize $Q_3$}};
\node []  at (-4.3,4.3) {{\footnotesize $Q_2$}};
\draw[](3.5,0)--(0,3.5);
\node [fill=white, inner sep=0pt]  at (1.99,1.84) {{\color{black}{\footnotesize $c_1$}}};
 \end{tikzpicture}
\hfill
\begin{tikzpicture}[scale=.7,>={MyStealth}]
\node []  at (-5.5,4.5) {(b)};  
\draw[step=.5cm,color=gray, ultra thin] (-4.2,-4.2) grid (4.2,4.2);
 \draw[->,black, ultra thick] (0,-4.3) -- (0, 4.3) node[above]{\footnotesize$y$};
  \draw[->,black, ultra thick] (-4.3,0) -- (4.3, 0) node[right]{\footnotesize$x$};
     \draw[black, dashed, thick] (0.5,0) -- (4.2, 3.7);
    \node []  at (4.3,4.3) {{\footnotesize $Q_1$}};
\node []  at (4.3,-4.3) {{\footnotesize $Q_4$}};
\node []  at (-4.3,-4.3) {{\footnotesize $Q_3$}};
\node []  at (-4.3,4.3) {{\footnotesize $Q_2$}};
    \node [r-sq] at (2,1.5) {}; 
 \node [cl-sq] at (3,1) {}; 
  \node [b-sq] at (0.5,3.5) {}; 
  \node [] at (0.8,3.8) {{{\footnotesize $c$}}};
   \node [cl-sq] at (3.5,2) {}; 
    \node [cl-sq] at (1.5,3) {}; 
            \node [cl-sq] at (2.5,4) {}; 
      \node [g-ci] at (-2,-2) {}; 
    \node [g-ci] at (-1,-3.5) {}; 
      \node [g-ci] at (-3,-1) {}; 
        \node [g-ci] at (-2.5,-4) {}; 
          \node [g-ci] at (-1,-1.5) {}; 
            \node [g-ci] at (-3.5,-3) {}; 
     \node [g-ci] at (2,-1) {}; 
       \node [g-ci] at (1,-3.5) {}; 
         \node [g-ci] at (3.5,-2) {}; 
           \node [g-ci] at (1.5,-0.5) {}; 
             \node [g-ci] at (2.5,-2.5) {}; 
               \node [bl-ci] at (0.5,1) {}; 
                \node [bl-ci] at (2,3.5) {}; 
           \node [bl-ci] at (4,2.5) {};
       \node [bl-ci] at (1.5,2.5) {}; 
         \node [bl-ci] at (3.5,4) {}; 
           \node [bl-ci] at (2.5,0.5) {}; 
       \node [bl-ci] at (-2.5,1.5) {}; 
       \node [bl-ci] at (-1.5,3.5) {}; 
         \node [bl-ci] at (-1,1) {}; 
           \node [bl-ci] at (-3,2) {}; 
            \node [bl-ci] at (-3.5,4) {};    
\draw[](3.5,0)--(0,3.5);
\node [fill=white, inner sep=0pt]  at (1.99,1.84) {{\color{black}{\footnotesize $c_1$}}};

 \end{tikzpicture}
\caption{An illustration of the proof of Lemma~\ref{lem:quadrant}, where the voters are circles and the candidates are squares (only candidates in quadrant $Q_1$ are shown). 
The candidate $c_1\in Q_1$ closest to the origin (see the thin solid line)
is shown in black. 
Case (a) concerns any gray candidate $c\in Q_1$ below the
dashed line, while Case (b) concerns any gray candidate
$c\in Q_1$ above the dashed line.
In Case (a), the voters in $Q_2\cup Q_3$, shown in white, prefer
candidate $c_1$ over the alternate candidate $c$.
In Case (b), the voters in $Q_3\cup Q_4$, again in white, prefer
candidate $c_1$ over the alternate candidate $c$.
}\label{fig:lem2}
\end{figure}

\item[(b)] $y-y_1 \ge x-x_1$. 
Then the analogous argument (via the symmetry in $x$ and $y$) applies with respect to voters of the form $v=(x_v, y_v)$ with $y_v\le 0$.
That is, any voter $v\in Q_3\cup Q_4$ prefers $c_1$ over $c$,
and these two quadrants contain at least half of the voters.
This is illustrated in Figure~\ref{fig:lem2}(b).
\qedhere
\end{enumerate}
\end{proof}

\begin{thm}\label{thm:DS4}
The set $S=\{c_1, c_2, c_3, c_4\}$ is a dominating set in the majority digraph, and hence a Condorcet winning set. 
\end{thm}
\begin{proof}
By Lemma~\ref{lem:quadrant}, each $c_i\in Q_i$, $1 \leq i \leq 4$, beats every other candidate in its quadrant if it exists. 
Thus every candidate outside of $S=\{c_1, c_2, c_3, c_4\}$ is
beaten by a candidate in $S$. Thus $S$ satisfies condition (\ref{eq:dom-set}) and so is a dominating set in the majority digraph.
\qedhere
\end{proof}

From Theorem~\ref{thm:DS4}, the Condorcet dimension is at most $4$.

\section{An Upper Bound on the Condorcet Dimension: Infinity Norm}\label{sec:infinity}

We now consider the {\em infinity norm} (or {\em supremum norm}).
Given two points ${\bf p}_1=(x_1, y_1)$ and ${\bf p}_2=(x_2, y_2)$ in two dimensions, the distance between 
them is $d^{\infty}({\bf p}_1, {\bf p}_2) = \max ( |x_1-x_2|, |y_1- y_2| )$.
The Manhattan norm and the infinity norm are closely related. Specifically, imagine that we rotate the axes of measurement by $45$ degrees (say, clockwise as in Figure~\ref{fig:transform}).
Let $D^{\infty}$ be the infinity norm
using these new axes of measurement.
The resulting distances are equivalent to the
Manhattan norm.

\begin{claim}\label{claim:rotation}
For any pair of points ${\bf p}_1$ and ${\bf p}_2$ in two dimensions, we have
$d^1({\bf p}_1, {\bf p}_2)= \sqrt{2}\cdot D^{\infty}({\bf p}_1, {\bf p}_2)$.
\end{claim}
\begin{proof}
Without loss of generality, we may assume 
that ${\bf p}_1={\bf 0}$ and that ${\bf p}_2$
lies in the positive quadrant.
Thus
$d^1({\bf 0}, {\bf p}_2) = |x_2|+|y_2| = x_2+y_2$. 
On the other hand, $D^{\infty}({\bf 0}, {\bf p}_2) = \frac{1}{\sqrt{2}}\cdot (x_2+y_2)$.
The claim follows. This proof is illustrated in Figure~\ref{fig:transform}.
\end{proof}

 \begin{figure}[h!] 
 \begin{tikzpicture}[scale=.7,>={MyStealth}]
\node []  at (-4,5.5) {{(a)}};  
\node []  at (6.9,5.5) {{(b)}};  
\draw[step=1cm,color=gray, ultra thin] (-4.2,-4.2) grid (4.2,4.2);
 \draw[->, black, ultra thick] (0,-4.3) -- (0, 4.3);
  \node []  at (0, 4.6) {{\footnotesize {$y$} }}; 
  \draw[->, black, ultra thick] (-4.3,0) -- (4.3, 0);
         \node []  at (4.5, 0) {{\footnotesize {$x$} }}; 
    \node [r-sq] at (2,2) {}; \node []  at (2.45,2.3) {{\color{black}{\footnotesize ${\bf p}_1$}}};
 \node [cl-sq] at (0,0) {}; \node []  at (0.3,-.3) {\footnotesize $\bf0$};
  \node [b-sq] at (-2,0) {}; \node []  at (-2.4,-.3) {{\color{black}{\footnotesize ${\bf p}_2$}}};
 \end{tikzpicture}
\hskip-6mm
\lower10mm\hbox{
 \begin{tikzpicture}[scale=.7,>={MyStealth}]
\begin{scope}[rotate=-45]
 \draw[step=1cm,color=gray, ultra thin] (-4.2,-4.2) grid (4.2,4.2);
 \draw[->, black, ultra thick] (0,-4.3) -- (0, 4.3);
   \node []  at (0, 4.6) {{\footnotesize $y$ }}; 
  \draw[->, black, ultra thick] (-4.3,0) -- (4.3, 0);
   \node []  at (4.5, 0) {{\footnotesize $x$ }};
  \end{scope}
      \node [r-sq] at (2,2) {}; 
      \node [fill=white, inner sep=0pt] at (2,2.5) {{\color{black}{\footnotesize ${\bf p}_1$}}};
 \node [cl-sq] at (0,0) {}; 
 \node []  at (0,-.5) {\footnotesize $\bf0$};
  \node [b-sq] at (-2,0) {};
  \node [fill=white, inner sep=0pt]  at (-2.4,-.3) {{{\footnotesize ${\bf p}_2$}}};
\end{tikzpicture}
 }
\caption{Using the Manhattan norm in (a), we have $d^1({\bf 0}, {\bf p}_1)=4, d^1({\bf 0}, {\bf p}_2)=2$ and $d^1({\bf p_1}, {\bf p}_2)=6$. Rotating the axis by 45 degrees and the applying the infinity norm gives $\hat{d}^\infty({\bf 0}, {\bf p}_1)= 2\sqrt{2}= \frac{1}{\sqrt{2}}\cdot 4, \hat{d}^\infty({\bf 0}, {\bf p}_2)=\sqrt{2}=\frac{1}{\sqrt{2}}\cdot 2$ and $\hat{d}^\infty({\bf p_1}, {\bf p}_2)= 3\sqrt{2}=\frac{1}{\sqrt{2}}\cdot 6$.}\label{fig:transform}
\end{figure}

\begin{thm}\label{thm:four-infinity}
In a $2$-dimensional metric space with infinity norm, the Condorcet dimension of any election is at most $4$.
\end{thm}
\begin{proof}
Take any election instance. By Claim~\ref{claim:rotation}, there is an equivalent election with the Manhattan norm
that induces identical preference lists.
By Theorem~\ref{thm:DS4}, this election has Condorcet dimension at most $4$.
\end{proof}

\section{A Lower Bound on the Condorcet Dimension}\label{sec:lower}

In this section, we study lower bounds on the Condorcet dimensions. Again, we focus on the case $D=2$ of metric spaces in two dimensions. Specifically, we prove that instances exist in two dimensions where a Condorcet winner does not exist. This result applies for any $p$-norm, including the infinity and Manhattan norms.
    
\begin{lem} \label{lem:LB}
  In a $2$-dimensional metric space under the $p$-norm, there exist instances with Condorcet dimension at least~$2$, for any $p\ge 1$.  
\end{lem}

\begin{proof}
Consider an election with three voters ${\bf v}_1=(9,0), {\bf v}_2=(0,9)$ and ${\bf v}_3=(-9,0)$ and three candidates
${\bf p}_1=(1,-1), {\bf p}_2=(8, 10)$ and ${\bf p}_3=(-9, 9)$.
Using the $p$-norm, for any $p\ge 1$, the preferences rankings of the voters
are 
${\bf v}_1: {\bf p}_1 \succ {\bf p}_2 \succ {\bf p}_3$, 
${\bf v}_2: {\bf p}_2 \succ {\bf p}_3 \succ {\bf p}_1$,
and
${\bf v}_3: {\bf p}_3 \succ {\bf p}_1 \succ {\bf p}_2$.
Hence there is a Condorcet cycle and so no Condorcet winner.
\end{proof}

\section{Embedding Voter Preferences in a High-Dimensional Metric Space}\label{sec:embed}

After our study of bounding and finding a Condorcet winning set for instances that can be represented metrically equipped with some $p$-norm, a natural question is the following: 
Can we test whether an instance has a representation in a metric space, and can we also find it? We answer these questions affirmatively. 
Already Bogomolnaia and Laslier~\cite{BL07} and Chen et al.~\cite{CNS22}
showed that any set of voter preferences over candidates has a corresponding embedding in a metric space with proximity voting, for the Euclidean norm ($p=2$) and Manhattan norm ($p=1$), respectively, for $D = \min\, [m-1, n]$.
We provide simple constructions for general $p$-norms, giving $D=\min [m, n]$, in the subsequent two theorems.

\begin{thm}\label{thm:simplex-m}
Consider an election with $m$ candidates and $n$ voters, each with a strict
preference list $\succ_v$ per voter~$v$ over the candidates. 
An embedding of dimension $D=m$ can be computed in polynomial time for any $p$-norm, $p\ge 0$.
\end{thm}
\begin{proof}
In our construction we place voter $j$ at position
${\bf y}_j$ and candidate $i$ at position ${\bf x}_i$. 
First, we place the $m$ candidates at the corners of a simplex in $m$ dimensions. Specifically, let ${\bf e}_i=(0,\dots, 0,1,0,\dots,0)$,
where the $i$th coordinate has value $1$. 
Then we position candidate $i$ at ${\bf x}_i = 2m\cdot {\bf e}_i$, for each $1\le i\le m$.
Next, consider voter $j$ for all $1\le j\le n$.
Let $\rho(i,j)$ be the position of candidate $i$ in the preference list of voter $j$. Place voter $j$ at ${\bf y}_j = (m-\rho(1,j), m-\rho(2,j), \dots, m-\rho(m,j))$. This embedding can be computed in polynomial time as it suffices to compute all the $\rho(i,j)$; this be done by reading the preference lists of each voter and their combined length is the input size of the instance.
 
Now, for the infinity norm we have
\begin{eqnarray*}
d^\infty({\bf y}_j, {\bf x}_i) &=& \max [ (2m-(m-\rho(i,j)) , \max_{\ell\neq i} (m-\rho(\ell, j) ]\\
&=& \max [ m+\rho(i,j)) , \max_{\ell\neq i} (m-\rho(\ell, j) ]\\
&=& m+\rho(i,j)
\end{eqnarray*}
Observe that these distances are increasing with $\rho(i,j)$. Hence
this positioning is consistent with the preference ranking of voter $j$.
Moreover this construction also works for any $p$-norm, with $0\le p <\infty$. To see this, observe that
\begin{eqnarray*}
&&d^p({\bf y}_j, {\bf x}_i)^p-d^p({\bf y}_j, {\bf x}_k)^p\\
&=& (m+\rho(i,j))^p + \sum_{\ell\neq i} \big(m-\rho(\ell, j)\big)^p
-(m+\rho(k,j))^p - \sum_{\ell\neq k} \big(m-\rho(\ell, j)\big)^p\\
&=& \Big((m+\rho(i,j))^p - (m+\rho(k,j)^p \Big)
+\Big( \big(m- \rho(k, j)\big)^p  -\big(m-\rho(i, j)\big)^p\Big) \\
&>& 0 
\end{eqnarray*}
where the strict inequality holds whenever $\rho(i,j)>\rho(k,j)$.
Thus the distances increase with $\rho(i,j)$.
\end{proof}

\begin{thm}\label{thm:simplex-n}
Consider an election with $m$ candidates and $n$ voters, each with a strict
preference list $\succ_v$ per voter~$v$ over the candidates. 
An embedding of dimension $D=n$ can be computed in polynomial time for any $p$-norm, $p>1$.
\end{thm}
\begin{proof}
In this construction, we place the $n$ voters, rather than the candidates,
at $n$ corners of an $n$-dimensional simplex. 
Specifically, we may use the construction of Bogomolnaia and Laslier~\cite{BL07}.
Place voter $j$ at ${\bf y}_j = B\cdot {\bf e}_j$, for each $1\le j\le n$, where $B$ (to be determined) is large.
Next we position the $m$ candidates. Again, let
$\rho(i,j)$ be the position of candidate $i$ in the preference list of voter $j$. Then we position each candidate~$i$ at 
${\bf x}_i = (-\rho(i,1), -\rho(i,2), \dots, -\rho(i,n))$.

For the infinity norm, we then have
\begin{eqnarray*}
d^\infty({\bf y}_j, {\bf x}_i) &=& \max\, \Big[\, B+\rho(i,j) \, ,\,  
 \max_{\ell\neq j} \rho(i, \ell)\big)\, \Big]\\
&=& B+\rho(i,j)
\end{eqnarray*}
Here the second equality holds if $B\ge m$.
These distances are then increasing with $\rho(i,j)$, so they are
consistent with the preference ranking of voter $j$. 

Next consider any $p$-norm, with $1< p <\infty$. Then
 \begin{eqnarray*}
&&d^p({\bf y}_j, {\bf x}_i)^p-d^p({\bf y}_j, {\bf x}_k)^p\\
&=& (B+\rho(i,j))^p + \sum_{\ell\neq j} \rho(i, \ell)^p
-(B+\rho(k,j))^p - \sum_{\ell\neq j} \rho(k, \ell)^p\\
&\le& (B+\rho(i,j))^p  -(B+\rho(k,j))^p  + (n-1)\cdot (m^p-1)\\
&<& 0 
\end{eqnarray*}
because, since $p>1$, there exists a sufficiently large choice of $B$ such that the strict inequality holds whenever $\rho(i,j) < \rho(k,j)$. 
Hence, the distances are increasing with $\rho(i,j)$ and are consistent with the preference ranking of voter $j$.   
\end{proof}

\section{Conclusion}
We presented bounds on the Condorcet dimension under 
the spatial model of voting. Several open problems remain. 
For the case ${D=2}$, we have shown the Condorcet dimension is at least~$2$ and at most~$4$; closing the gap between the lower and upper bounds is an intriguing problem.
Our upper bound holds for the Manhattan and infinity norms;
does it extend to other distance norms? 
However, the outstanding open problem concerns 
whether or not the Condorcet dimension is at most $3$ in the general case (equivalently, instances where $D$ can be arbitrarily large). 
One current evidence for this is an observation by Bloks \cite{Bloks18}
that elections where the 
majority digraph
defines a tournament with minimum dominating set
of size four seem to be ``random'' enough to have very small
Condorcet dimension. 

\

\noindent{\bf Acknowledgement.} 
We are very grateful to Louis-Roy Langevin for finding an error in a previous version of the paper that claimed a factor 3 upper bound.
We thank the referees for their thoughtful comments and suggestions to improve this paper and Jannik Peters for comments and directing us to the papers \cite{BL07} and \cite{CNS22}. 
This work was supported by NSERC grant 2022-04191.

\bibliography{sample}

\end{document}